\RequirePackage{lineno}

\documentclass[envcountsame]{llncs}
\RequirePackage{bussproofs/bussproofs}
\usepackage{xspace}
\usepackage{subfigure}
\usepackage{amsmath,amssymb}
\usepackage{enumerate}
\usepackage{wrapfig}
\usepackage{cite}
\usepackage[usenames,dvipsnames]{color}

\usepackage{graphicx}
\usepackage{xspace}
\usepackage{hyperref} 

\usepackage{marvosym}

\definecolor{gray}{gray}{0.5}

\mathchardef\mhyphen="2D

\newlength{\breite}
\newcommand{\class}[1]{\ensuremath{\mathsf{#1}}}
\DeclareMathOperator*{\lev}{\operatorname{lv}}
\DeclareMathOperator*{\var}{\operatorname{var}}
\DeclareMathOperator*{\ind}{\operatorname{ind}}
\DeclareMathOperator*{\range}{{\sf rng}}
\DeclareMathOperator*{\domain}{{\sf dom}}
\DeclareMathOperator*{\instantiate}{\textsf{inst}}
\DeclareMathOperator*{\rest}{{\sf restrict}}
\renewcommand{\restriction}{\mathord{\upharpoonright}}
\newcommand{\comprehension}[2]{\ensuremath{\left\{ {#1} \;|\; {#2}\right\}}}
\newcommand{\fovar}[1]{{#1}}
\newcommand{\fopre}[1]{{#1}}
\newcommand{\qrc}{\textsf{Q-Res}\xspace}
\newcommand{\qurc}{\textsf{QU-Res}\xspace}
\newcommand{\lqrc}{\textsf{LD-Q-Res}\xspace}
\newcommand{\lquprc}{\textsf{LQU}$^+$\textsf{-\hspace{1pt}Res}\xspace}
\newcommand{\irmc}{\textsf{IRM-calc}\xspace}
\newcommand{\ecalculus}{$\forall$\textsf{Exp+Res}\xspace}
\newcommand{\lqurc}{\textsf{LQU}\textsf{-Res}\xspace}
\newcommand{\epr}{\textsf{EPR}\xspace}
\newcommand{\for}{\textsf{FO-res}\xspace}
\newcommand{\irc}{\textsf{IR-calc}\xspace}
\newcommand{\dirc}{\textsf{D-IR-calc}\xspace}

\newenvironment{nscenter}
 {\parskip=0pt\par\nopagebreak\centering}
 {\par\noindent\ignorespacesafterend}

\usepackage{tikz}
\usetikzlibrary{fadings,patterns}
\usepgflibrary{arrows}
\tikzstyle{uedge}=[draw=blue!50!red]
\tikzstyle{fedge}=[draw=blue]
\tikzstyle{iedge}=[draw=red]
\tikzstyle{redge}=[draw=green!50!black]
\tikzstyle{rnode}=[draw,inner sep=2pt,color=black]
\tikzstyle{tnode}=[circle,minimum width=3pt,fill,inner sep=0pt]
\tikzstyle{dotnode}=[circle,minimum width=2pt,fill,inner sep=0pt]
\tikzstyle{labn}=[font=\sffamily,circle,fill=white,inner sep=1pt,draw=black]
\tikzstyle{legn}=[font=\scriptsize]
\tikzstyle{reln}=[circle,fill=white,inner sep=.4pt,draw=black]
\tikzstyle{oreln}=[circle,fill=white,inner sep=.4pt,draw=black!50,solid]
\tikzstyle{oree}=[thick,draw=black!50,densely dashed]
\tikzstyle{ree}=[thick,draw=black]

\tikzstyle{calcn-incomplete}=[rectangle%
,rounded corners=2mm%
,thick%
,dashed
,draw=blue
,minimum height=.5cm,minimum width=.5cm,inner sep=2pt%
]
\tikzstyle{expcalcn}=[rectangle%
,thick%
,draw=green!50!black,%!60%
minimum height=.5cm,minimum width=.5cm,inner sep=2pt%
]

\tikzstyle{calcn-unsound}=[rectangle%
,rounded corners=2mm%
,line width=2pt%
,draw=red,
minimum height=.5cm,minimum width=.5cm,inner sep=2pt%
]

\tikzstyle{expcalcn-unsound}=[rectangle%
,line width=2pt%
,draw=red
,minimum height=.5cm,minimum width=.5cm,inner sep=2pt%
]

\tikzstyle{calcn-nocolour}=[rectangle%
,rounded corners=2mm%
,thick%
,draw=black,%
minimum height=.5cm,minimum width=.5cm,inner sep=2pt%
]
\tikzstyle{expcalcn-nocolour}=[rectangle%
,thick,%
,draw=black,%
minimum height=.5cm,minimum width=.5cm,inner sep=2pt%
]

\title{Lifting QBF Resolution Calculi to DQBF}  %Can QBF Resolution be lifted to DQBF?}  %What is the right Resolution system for DQBF?}    %A Resolution System for DQBF}
\author{Olaf Beyersdorff\inst{1} \and Leroy Chew\inst{1} \and Renate A. Schmidt\inst{2} \and Martin Suda\inst{2} \email{o.beyersdorff@leeds.ac.uk, mm12lnc@leeds.ac.uk, Renate.Schmidt@manchester.ac.uk, martin.suda@manchester.ac.uk}}
\institute{
School of Computing, University of Leeds, UK
\and
School of Computer Science, University of Manchester, UK 
}
\begin{document}
\maketitle

\begin{abstract}
  We examine existing resolution systems for quantified Boolean formulas (QBF) and answer the question which of these calculi can be lifted to the more powerful Dependency QBFs (DQBF). An interesting picture emerges: While for QBF we have the strict chain of proof systems $\qrc < \irc < \irmc$, the situation is quite different in DQBF. The obvious adaptations of \qrc and likewise universal resolution are too weak: they are not complete. The obvious adaptation of \irc has the right strength: it is sound and complete. \irmc is too strong: it is not sound any more, and the same applies to long-distance resolution.     
Conceptually, we use the relation of DQBF to effectively propositional logic (\epr) and explain our new DQBF calculus based on \irc as a subsystem of first-order resolution.  
\end{abstract}

\section{Introduction}
The logic of dependency quantified Boolean formulas (DQBF) \cite{Peterson79} generalises the notion of quantified Boolean formulas (QBF) that allow Boolean quantifiers over a propositional problem.
DQBF is a relaxation of QBF in that the quantifier order is no longer necessarily linear and the dependencies of the quantifiers are completely specified.
This is achieved using Henkin quantifiers \cite{Henkin61}, usually put into a Skolem %or Herbrand 
form.
DQBF is \class{NEXPTIME}-complete \cite{Azhar01lowerbounds},
compared to the \class{PSPACE}-completeness of QBF \cite{DBLP:conf/stoc/StockmeyerM73}.
Thus, unless the classes are equal, many problems that are difficult to express in QBF can be succinctly represented in DQBF. 

Recent developments in QBF proof complexity \cite{BWJ14,JM15,BCJ14,BCJ15,BCMS15,BBC16,BCMS16,DBLP:conf/sat/Egly12,Slivovsky-sat14, heule2014unified, heule2014efficient}
have increased our theoretical understanding of QBF proof systems and proof systems in general and have shown that there is an intrinsic link between proof complexity and SAT-, QBF-, and DQBF-solving.
Lower bounds in resolution-based proof systems give lower bounds to CDCL-style algorithms.
In propositional logic there is only one resolution system (although many subsystems have been studied \cite{Seg07,Urq95}),
but in QBF, resolution can be adapted in  different ways to get sound and complete calculi of varying strengths \cite{BCJ14,JM15,Gelder12,DBLP:series/faia/GiunchigliaMN09}. 

The first and best-studied QBF resolution system is \qrc introduced in \cite{DBLP:journals/iandc/BuningKF95}.
For \qrc there are two main enhanced versions: \qurc~\cite{Gelder12},
which allows resolution on universal variables, and \lqrc~\cite{DBLP:series/faia/GiunchigliaMN09},
which introduces a process of merging positive and negative universal literals under certain conditions.
These two concepts were combined into a single calculus \lquprc \cite{BWJ14}.

While these calculi model CDCL solving, another group of resolution systems were developed with the goal to express ideas from expansion solving. The first and most basic of these systems is \ecalculus \cite{JM15}, which also uses resolution, but is conceptually very different from \qrc.
In \cite{BCJ14} two further proof systems \irc and  \irmc are introduced,
which unify the CDCL- and expansion-based approaches in the sense that \irc simulates both \qrc and \ecalculus.
The system \irmc enhances \irc and additionally simulates \lqrc. 
The relative strength of these QBF resolution systems is illustrated in Fig.~\ref{fig:overview}. 

\begin{figure}[t]
\hspace*{-2cm}
\begin{tikzpicture}[scale=1.25]
\node[calcn-incomplete](n1) at (3,1) {{\sf Tree-}\qrc} ;
\node[calcn-incomplete](n2) at (3,2) {\qrc} ;
\node[expcalcn](n3) at (0,2) {\ecalculus} ;
\node[calcn-unsound](n4) at (2,3) {\lqrc} ;
\node[calcn-incomplete](n5) at (4,3) {\qurc} ;
\node[calcn-unsound](n6) at (3,4) {\lquprc} ;
\node[expcalcn](n8) at (0,3) {\irc} ;
\node[expcalcn-unsound](n9) at (0,4) {\irmc} ;
\draw(n4)--(n2);
\draw(n6)--(n4);
\draw(n2)--(n8);
\draw(n8)--(n3);
\draw (n3)--(n1)--(n2)--(n5);
\draw(n5)--(n6);
\draw(n8)--(n9)--(n4);
\draw[densely dashed](n2)--(n3);
\draw[densely dashed](n4)--(n5);
\draw[densely dashed](n8)--(n4);
\draw[densely dashed](n3.west)
.. controls (-2.5,4) and (0,5.5) .. (n6);

\draw[densely dashed](n5.north) .. controls (4.8,5) and (1,5.5) ..   (n9);

%\node[densely dashed,labn](l10) at (-1,4.25) {10};
%\draw[densely dashed,very thick](n3.west) to[bend left] (l10);
%\draw[densely dashed,very thick](l10) to[bend left=10] (n6);
%\draw[densely dashed,very thick](n3) to (-1,2.5) to (-1,4.5)to node[labn]{10}(2,4.5) to (n6);

%\draw [preaction={,top color=black!60,bottom color=black!30 ,transform canvas={xshift=1mm,yshift=-1mm}}]
%      [fill=white,thick] (4.65,4.16) rectangle (9,1.75);
\draw (4.75,4.75)--(5.1,4.75); \node[legn,right] at(5.25,4.75) {\footnotesize{strictly stronger (p-simulates,}};
\node[legn,right] at(5.25,4.5) {\footnotesize{but exponentially separated)}};
\draw[densely dashed] (4.75,4)--(5.1,4); \node[legn,right] at(5.25,4) {\footnotesize{incomparable (mutual }};
\node[legn,right] at(5.25,3.75) {\footnotesize{exponential separations)}};
%\draw[very thick] (4.75,2.6)--(5.1,2.6); 
%\draw[very thick,densely dashed] (4.75,2.4) --(5.1,2.4);
%\node[legn,right] at(5.25,2.5) {\footnotesize{new results}};
\node[expcalcn-nocolour,minimum height=4mm] at (4.925,3.15) {} ; \node[legn,right] at(5.25,3.15) {\footnotesize{expansion solving}};
\node[calcn-nocolour,minimum height=4mm] at (4.925,2.65) {} ; \node[legn,right] at(5.25,2.65) {\footnotesize{CDCL solving}};

%\draw[pattern=north west lines, pattern color=blue] (0,0) rectangle (2,4);
%\draw[pattern=dots] 
\draw[red, line width=2pt](4.925,2) circle (1.5mm); \node[legn,right] at(5.25,2) {\footnotesize{unsound in DQBF}};

\draw[thick, color=green!50!black] (4.925,1.5) circle (1.5mm);
%\fill[color=green!50] (4.925,1.5) circle (1.5mm); %\node[expcalcn,minimum height=4mm] at (4.925,1.5) {} ;
\node[legn,right] at(5.25,1.5) {\footnotesize{sound and complete in DQBF}};

\draw[thick, color=blue, dashed] (4.925,1) circle (1.5mm); %\node[calcn,minimum height=4mm] at (4.925,1) {} ; 
\node[legn,right] at(5.25,1) {\footnotesize{incomplete in DQBF}};

\end{tikzpicture}
	\caption{The simulation order of QBF resolution systems \cite{BCJ15} and soundness/completeness of their lifted DQBF versions} 
	\label{fig:overview}
\end{figure}
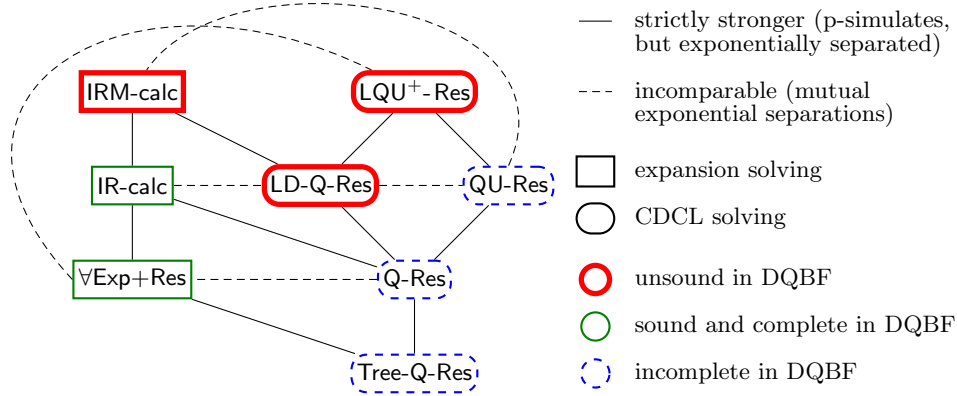

The aim of this paper is to clarify which of these QBF resolution systems can be lifted to DQBF. This is motivated both by the theoretical quest to understand which QBF resolution paradigms are robust enough to work in the more powerful DQBF setting, as well as from the practical perspective, where recent advances in DQBF solving \cite{DBLP:conf/sat/WimmerGNSB15,DBLP:conf/date/GitinaWRSSB15,DBLP:conf/sat/FinkbeinerT14,DBLP:conf/sat/FrohlichKBV14} prompt the question of how to model and analyse these solvers proof-theoretically.

Our starting point is the work of Balabanov, Chiang, and Jiang \cite{Balabanov201486}, who show that \qrc can be naturally adapted to a sound calculus for DQBF, but they show it is not strong enough and lacks completeness. Using an idea from \cite{BWJ14} we extend their result to \qurc, thus showing that the lifted version of this system to DQBF is not complete either.
We present an example showing that the lifted version of \lqrc is not sound, and this transfers to the DQBF analogues of the stronger systems \lquprc and \irmc.

While this rules out most of the existing QBF resolution calculi already---and in fact all CDCL-based systems (cf.\ Fig.~\ref{fig:overview})---we show that \irc, lifted in a natural way to a DQBF calculus \dirc, is indeed sound and complete for DQBF; and this holds as well for the lifted version of the weaker expansion system \ecalculus.

Conceptually, our soundness and completeness arguments use the known correspondence of QBF and DQBF to first-order logic \cite{PAAR-2012:qbf2epr_A_Tool_for_Generating_EPR_Formulas_from_QBF},
and in particular to the fragment \epr, also known as the Bernays-Sch\"{o}nfinkel class,
the universal fragment of first-order logic without function symbols of non-zero arity.
Similarly to DQBF, EPR is \class{NEXPTIME}-complete \cite{LEWIS1980317}. 
In addition to providing soundness and completeness this explains the semantics of both 
\irc and \dirc and identifies these systems as special cases of first-order resolution.

\section{Preliminaries}

A {\em literal} is a Boolean variable or its negation.
%we say that the literal $x$ is \emph{complementary} to the literal $\lnot x$ and vice versa.
If~$l$ is a literal, $\lnot l$ denotes the complementary literal,
i.e., $\lnot \lnot x=x$.
A {\em clause} is a set of literals understood as their disjunction.
The empty clause is denoted by~$\bot$,
which is semantically equivalent to false.
A formula in {\em Conjunctive Normal Form} (CNF) is a
conjunction of clauses.  
For a literal $l=x$ or $l=\lnot x$, we write $\var(l)$ for~$x$ and extend this notation to $\var(C)$ for a clause $C$. % and

A \emph{Dependency Quantified Boolean Formula (DQBF)} $\phi$ in prenex Sko\-lem form consists of a quantifier prefix $\Pi$ and a propositional matrix $\psi$.
QBF $\phi$ can also be written as $\Pi\cdot \psi$.
Here we mainly study DQBFs where $\psi$ is in CNF. The propositional variables of $\psi$ are partitioned into sets $Y$ and $X$.
We define $Y$ as the set of universal variables and $X$ the set of existential variables.
For every $ y\in Y$, $\Pi$ contains the quantifier $\forall y$.
For every $x\in X$ there is a predefined subset $Y_x\subseteq Y$ and 
$\Pi$ contains the quantifier $\exists x(Y_x)$.

The semantics of DQBF is defined in terms of Skolem functions.
A Skolem function $f_x: \{0,1\}^{Y_x} \rightarrow \{0,1\}$ 
describes the evaluation of an existential variable $x$ under the possible assignments to its dependencies $Y_x$.
Given a set $F$ of Skolem functions, where  $F=\comprehension{f_x}{x\in X}$ for all the existential variables 
and an assignment $\alpha : Y \rightarrow \{0,1\}$ for the universal variables,
the \emph{extension} of $\alpha$ by~$F$ is defined as $\alpha_F(x) = f_x(\alpha \restriction Y_x)$
for $x\in X$ and $\alpha_F(y) = \alpha(y)$ for $y \in Y$.
A DQBF $\phi$ is \emph{true} if there exist Skolem functions $F=\comprehension{f_x}{x\in X}$
for the existential variables such that for every assignment $\alpha : Y \rightarrow \{0,1\}$
to the universal variables the matrix $\psi$ propositionally evaluates to $1$
under the extension $\alpha_F$ of~$\alpha$ by $F$.

In QBF, the quantifier prefix is a sequence of standard quantifiers of the form $\exists x$ and $\forall y$. 
To see that this is a special case of DQBF, we use the sequence from left to right 
to assign to every variable in the prefix a unique index $\ind:X\cup Y\rightarrow \mathbb{N}$, 
and make every existential variable $x$ depend on all the preceding 
universal variables by setting $Y_x=\{y\in Y\mid \ind(y)<\ind(x)\}$.

\iffalse
% we don't seem to need level!
In addition, we can also define an index on the variables which imposes a total ordering on all the variables such that $\lev(a)<\lev(b)$ implies $\ind(a)<\ind(b)$.
The index does not affect semantics but can be utilised in a proof system.
We can always take the view that the index may be picked arbitrarily in such a proof system.

% repeated from the intro:
Deciding truth in DQBF is \class{NEXPTIME}-complete, while deciding truth in QBF is \class{PSPACE}-complete.

A \emph{proof system} \cite{CR79} for a
language $L$ over $\Gamma$ is a poly-time computable %partial
function $f:\Gamma^\star\to\Gamma^\star$ with $\mathit{rng}(f)=L$.
%and where $\Gamma^*$ is the set of strings over $\Gamma$.
If $f(x)=y$ then $x$ is called an $f$-proof for $y$.
%If $L$ consists of all propositional tautologies, then $f$ is called a \emph{propositional proof system}. 
For $L=\text{QBF}$ we speak of a \emph{QBF proof system}.
In our systems here, proofs are sequences of clauses; a \emph{refutation} is a proof deriving $\bot$.
A proof system $S$ for $L$ {\em simulates} a proof system $P$ for $L$
if there exists a polynomial $p$ such that for
all $P$-proofs $\pi$ of $x$ there is an
$S$-proof $\pi^\prime$ of $x$ with $|\pi^\prime|\leq p\left( |\pi|\right) $.
% If such a proof $\pi^\prime$ can even be computed from $\pi$ in polynomial time we say that $S$ {\em p-simulates} $P$.
\fi

\smallskip
% \textbf{Resolution-based calculi for QBF.}
We now give a brief overview of the main existing resolution-based calculi for QBF.
We start by describing the proof systems modelling 
\emph{CDCL-based QBF solving};
their rules are summarized in Fig.~\ref{fig:allrules}. The most basic and important 
system is \emph{Q-resolution (\qrc)} by
Kleine B{\"u}ning~et~al.~\cite{DBLP:journals/iandc/BuningKF95}. It is a resolution-like
calculus 
that operates on QBFs in prenex form with CNF matrix. In addition to the axioms,
\qrc comprises the resolution rule S$\exists$R and universal reduction $\forall$-Red (cf.\ Fig.~\ref{fig:allrules}).

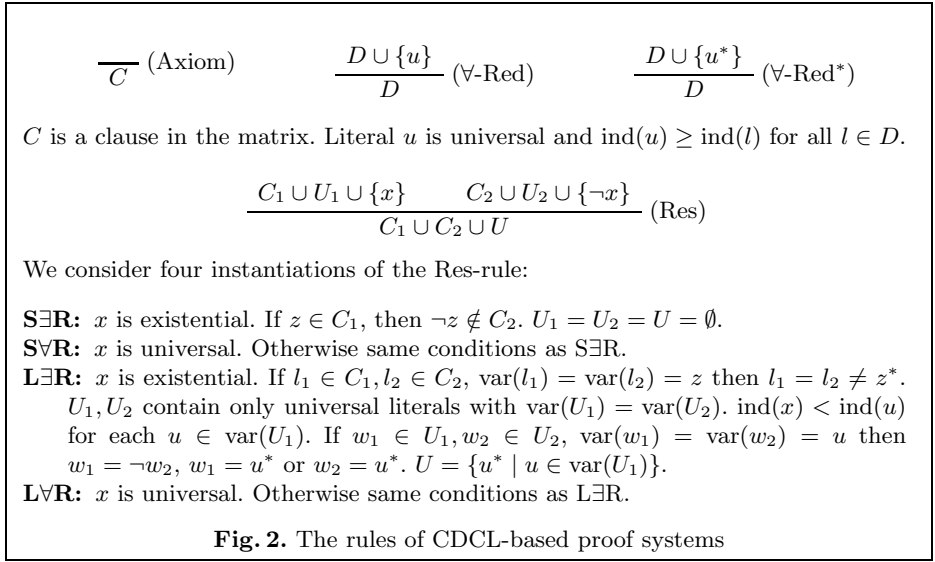
\emph{Long-distance resolution (\lqrc)} appears originally in the work of Zhang and
Malik~\cite{DBLP:conf/iccad/ZhangM02}
and was formalized into a calculus by Balabanov and Jiang~\cite{DBLP:journals/fmsd/BalabanovJ12}.
It merges complementary literals of a universal variable~$u$
into the special literal~$u^*$.
These special literals prohibit certain resolution steps.
In particular, different literals of a universal variable~$u$
may be merged only if $\ind(x)<\ind(u)$, where~$x$ is the resolved variable.
\lqrc uses the rules  L$\exists$R, $\forall$-Red and $\forall$-Red$^*$. 

%The rules are given in \autoref{fig:LD}.
%Note that the rules do not prohibit resolving $w^*\lor x\lor C_1$ and $u^*\lor\lnot x\lor C_2$ with $\lev(w)\leq\lev(u)<\lev(x)$ as long as $w\neq u$.

\emph{QU-resolution (\qurc)} \cite{Gelder12} removes the restriction from \qrc 
that the resolved variable must be an existential variable and allows resolution of universal variables.
The rules of \qurc are  S$\exists$R, S$\forall$R and $\forall$-Red. 
%given in \autoref{fig:QU}.
%
\emph{\lquprc} \cite{BWJ14} extends \lqrc by allowing short and long distance resolved literals to be universal;
however, the resolved literal is never a merged literal $z^*$. \lquprc uses the rules L$\exists$R, L$\forall$R, $\forall$-Red and $\forall$-Red$^*$. 

\begin{figure}[t]
\framebox{
	\parbox{\breite}
	{\small
		\begin{prooftree}
			\AxiomC{}
			\RightLabel{(Axiom)}
			\UnaryInfC{$C$}
			\DisplayProof\hspace{1cm}
			\AxiomC{$D\cup\{u\}$}
			\RightLabel{($\forall$-Red)}
			\UnaryInfC{$D$}
			\DisplayProof\hspace{1cm}
			\AxiomC{$D\cup\{u^*\}$}
			\RightLabel{($\forall$-Red$^*$)}
			\UnaryInfC{$D$}
		\end{prooftree}
		\begin{minipage}{0.99\linewidth}
			$C$ is a clause in the matrix. Literal $u$ is universal and $\ind(u)\geq\ind(l)$ for all $l\in D$.
		\end{minipage}
		\begin{prooftree}
			\AxiomC{$C_1\cup U_1\cup\{x\}$}
			\AxiomC{$C_2 \cup U_2\cup\{\lnot{x}\}$}
			\RightLabel{(Res)}
			\BinaryInfC{$C_1\cup C_2\cup U$}
		\end{prooftree}
		\begin{minipage}{0.99\linewidth}
			We consider four instantiations of the Res-rule:
\begin{description}
			\item[S$\exists$R:]  $x$ is existential. If $z\in C_1$, then $\lnot{z}\notin C_2$. $U_1=U_2=U=\emptyset$.
			\item[S$\forall$R:]  $x$ is universal. Otherwise same conditions as S$\exists$R.
			\item[L$\exists$R:]  $x$ is existential. If  $l_1\in C_1, l_2\in C_2$, $\var(l_1)=\var(l_2)=z$ then $l_1=l_2\neq z^*$.
			$U_1, U_2$ contain only universal literals with $\var(U_1)=\var(U_2)$.
			$\ind(x)<\ind(u)$ for each $u\in\var(U_1)$.
			If $w_1\in U_1, w_2\in U_2$, $\var(w_1)=\var(w_2)=u$ then $w_1=\lnot w_2$, $w_1=u^*$ or $w_2=u^*$. $U=\{u^* \mid u\in \var(U_1)\}$.
			\item[L$\forall$R:]  $x$ is universal. Otherwise same conditions as L$\exists$R.
\end{description}
		\end{minipage}
		\caption{The rules of CDCL-based proof systems}
		\label{fig:allrules}
	}
	}
	
\end{figure}

\smallskip
% \textbf{Expansion-based QBF solving}
The second type of calculi models \emph{expansion-based QBF solving}. These calculi are 
based on \emph{instantiation} of universal variables:  
\ecalculus~\cite{JanotaSilva-SAT13}, \irc, and \irmc~\cite{BCJ14}.  All these
calculi operate on clauses that comprise only existential variables from the original QBF, which
are additionally \emph{annotated} by a substitution to some universal variables, e.g.\ $\lnot
x^{0/u_1 1/u_2}$.  
For any annotated literal $l^\sigma$, the substitution $\sigma$ must not make
assignments to variables at a higher quantification index than that of $l$, i.e., if
$u\in\domain(\sigma)$, then $u$ is universal and $\ind(u)<\ind(l)$.

To preserve this invariant we use the following definition. 
%\begin{definition}\label{def:rest_QBF}
\iffalse
	Fix a QBF $\Pi\cdot \psi$. Let $\tau$ be a partial assignment of the universal variables $Y$ to $\{0,1\}$
	and let $x$ be an existential variable. Then $\rest_x(\tau)$ is the assignment 
	where $\domain(\rest_x(\tau))= \{u\in \domain(\tau)\mid \ind(u)<\ind(x)\}$ and $\rest_x(\tau)(u)=\tau(u)$. 
\fi
	Fix a DQBF $\Pi\cdot \psi$. Let $\tau$ be a partial assignment of the universal variables $Y$ to $\{0,1\}$
	and let $x$ be an existential variable. $\rest_x(\tau)$ is the assignment where $\domain(\rest_x(\tau))= \domain(\tau)\cap Y_x$ and $\rest_x(\tau)(u)=\tau(u)$. 
	
%\end{definition}

%To preserve this invariant, we use the \emph{auxiliary notation $l^{[\sigma]}$}, which for an existential literal $l$  and an assignment $\sigma$ to the universal variables 
%filters out all assignments that are not permitted,
%i.e.\ $l^{[\sigma]}=l^{\comprehension{c/u\in\sigma}{\lev(u)<\lev(l)}}$.
%i.e.\ $l^{[\sigma]}=l^{\sigma'}$ with
%$\sigma'=\comprehension{c/u\in\sigma}{\lev(u)<\lev(l)}$.

The simplest of the instantiation-based calculi we consider is \ecalculus from \cite{JM15} (cf.\ also \cite{BCJ14,BCJ15}).
%
\iffalse
\begin{figure}
	\framebox{\parbox{\textwidth}
		{\small
			\begin{prooftree}
				\AxiomC{}
				\RightLabel{(Ax)}
				\UnaryInfC{$
					\{l^{\rest_x(\tau)}\mid l\in C, l\text{ exist.}\}
					\cup
					\{\tau(l)\mid l\in C, l\text{ univ.}\}
					$}
			\end{prooftree}
			%%%
			\begin{minipage}{0.99\linewidth}
				$C$ is a clause from the matrix and
				$\tau$ is an assignment to all universal variables.
			\end{minipage}
			%%%
			%\vspace{3pt}
			\begin{prooftree}
				\AxiomC{$C_1\lor x^{\tau}$}
				\AxiomC{$C_2\lor\lnot x^{\tau}$}
				\RightLabel{(Res)}
				\BinaryInfC{$C_1\cup C_2$}
			\end{prooftree}
			\caption{The rules of \ecalculus \cite{JanotaSilva-SAT13}}\label{fig:EC}
		}}
		
	\end{figure}
\fi
%Any axiom in a proof is taken from the matrix by choosing a complete assignment to the universal variables. Resolution is defined as in the propositional case where annotated literals are considered as distinct variables.
%
The system \irc extends \ecalculus by enabling partial assignments in annotations. To 
do so, we utilize the auxiliary operation of \emph{instantiation}.
%\begin{definition}\label{def_inst_qbf}
\iffalse
	We define $\instantiate_\tau(C)$ to be the clause containing all the literals $l^{\rest_{\var(l)}(\sigma)}$,
	where $l^\xi\in C$ and $\domain(\sigma)=\domain(\xi)\cup \domain(\tau)$ and $\sigma(u)=\xi(u)$
	if $u\in \domain(\xi) $ and $\sigma(u)=\tau(u)$ otherwise.
	\fi
		We define $\instantiate_\tau(C)$ to be the clause containing all the literals $l^{\rest_{\var(l)}(\sigma)}$, 
		where $l^\xi\in C$ and $\domain(\sigma)=\domain(\xi)\cup \domain(\tau)$ and $\sigma(u)=\xi(u)$ if $u\in \domain(\xi) $ and $\sigma(u)=\tau(u)$ otherwise.
	
%\end{definition}

	\begin{figure}[t]
		\framebox{\parbox{\breite}
			{%\small
				\begin{prooftree}
					\AxiomC{}
					\RightLabel{(Axiom)}
					\UnaryInfC{$\comprehension{x^{\rest_x(\tau)}}{x\in C, x\text{ is existential}} $}
				\end{prooftree}
				$C$ is a non-tautological clause from the matrix. $\tau=\comprehension{0/u}{u\text{ is universal in }C}$, where the notation $0/u$ for literals $u$ is shorthand for $0/x$ if $u=x$ and $1/x$ if $u=\neg x$.
				\begin{prooftree}
				\AxiomC{$\{x^\tau\}\cup C_1 $ }
				\AxiomC{$\{\lnot x^\tau\}\cup C_2 $}
					\RightLabel{(Resolution)}
					\BinaryInfC{$C_1\cup C_2$}
					%\end{prooftree}
					\DisplayProof\hspace{1cm}
					%\begin{prooftree}
					\AxiomC{$C$}
					\RightLabel{(Instantiation)}
					\UnaryInfC{$\instantiate_\tau(C)$}
				\end{prooftree}
				\text{$\tau$ is an assignment to universal variables with $\range(\tau) \subseteq \{0,1\}$.}
				\caption{ The rules of \irc \cite{BCJ14} and of \dirc (Section~\ref{sec:dirc})}
				\label{fig:IRC}
			}}
		\end{figure}

  %For assignments
%$\tau$ and $\mu$, we write $\complete{\tau}{\mu}$ for the assignment $\sigma$ defined as follows:
%$\sigma(x)=\tau(x)$ if $x\in\domain(\tau)$,  
%otherwise $\sigma(x)=\mu(x)$ if $x\in\domain(\mu)$.
%The operation $\complete{\tau}{\mu}$ is called \emph{completion} because $\mu$
%provides values for variables not defined in $\tau$.  The operation is associative and therefore
%we can omit parentheses.  For an assignment $\tau$ and
%an annotated clause $C$ the function $\instantiate(\tau,C)$ returns the annotated clause
%$\comprehension{l^{[\complete{\sigma}{\tau}]}}{l^\sigma\in C}$.  The system \irc is
%defined in \autoref{fig:IRC}.  
%Axioms are taken from the matrix by assigning only those universal variables that appear in that  matrix clause. Any clause can be instantiated by giving values to some universal variables (this is similar to \emph{specialization} in first-order logic). Resolution is defined as in \ecalculus.
%Two examples on derivations in \ecalculus and \irc are contained in the appendix (Fig.~\ref{fig:proofexs}).

The calculus \irmc from \cite{BCJ14} further extends \irc by enabling annotations containing $*$, similarly as in \lqrc.  %We call such assignments  \emph{extended assignments}.  

\iffalse
The rules of the
calculus \irmc are presented in Fig.~\ref{fig:IRMC}.  The symbol $*$  may be introduced by the merge
rule, e.g.\ by collapsing $x^{0/u}\lor x^{1/u}$ into $x^{*/u}$.

\begin{figure}
\framebox{\parbox{\textwidth}
	{\small
		Axiom and instantiation rules as in \irc in Fig.~\ref{fig:IRC}.
		\begin{prooftree}
			\AxiomC{$x^{\tau\cup\xi}\lor C_1 $ }
			\AxiomC{$\lnot x^{\tau\cup\sigma}\lor C_2 $}
			\RightLabel{(Res)}
			\BinaryInfC{$\instantiate(\sigma,C_1)\cup\instantiate(\xi,C_2)$}
		\end{prooftree}
		\text{$\domain(\tau)$, $\domain(\xi)$ and $\domain(\sigma)$ are}
		\text{mutually disjoint. $\range(\tau)=\{0,1\}$}
		\begin{prooftree}
			\AxiomC{$C\lor b^\mu\lor b^\sigma$}
			\RightLabel{(Merging)}
			\UnaryInfC{$C\lor b^\xi$}
		\end{prooftree}
		\text{$\domain(\mu)=\domain(\sigma)$.}
		\text{$\xi=\comprehension{c/u}{c/u\in\mu,c/u\in\sigma}\cup \comprehension{*/u}{c/u\in\mu,d/u\in\sigma,c\neq d}$.}
		\caption{The rules of \irmc \cite{BCJ14} }\label{fig:IRMC}
	}}
\end{figure}
\fi

\section{Problems with lifting QBF calculi to DQBF}

There is no unique method for lifting calculi from QBF to DQBF.
However, we can consider `natural' generalisations of these calculi,
where we interpret index conditions as dependency conditions.
This means that when a proof system requires for an existential variable $x$ and a universal variable $y$
with $\ind(y)<\ind(x)$, this should be interpreted as $y\in Y_x$. Analogously $\ind(x)<\ind(y)$ should be interpreted as $y\notin Y_x$.
This approach was followed when taking Q-Resolution to D-Q-Resolution in Theorem~7 of \cite{Balabanov201486}. 
Balabanov et al. showed there that D-Q-Resolution is not complete for DQBF using some specific formula.
%It also follows that the tree-like version of D-Q-Resolution is not complete for DQBF.
%
%\section{From \epr to \irc}\label{sec:eprir}
%
%\begin{equation} 
%\label{eq:counterex1}
%\forall x_1 \forall x_2 \exists y_1(x_1) \exists y_2(x_2)\end{equation}
%\begin{nscenter}
%\begin{tabular}{lll}
%$\{y_1, y_2, x_1\}$&\hspace{0.2cm}&$ \{\neg y_1, \neg y_2, x_1\}$\\
%$\{y_1, y_2, \neg x_1, \neg x_2\}$&\hspace{0.2cm} &$\{\neg y_1, \neg y_2, \neg x_1, \neg x_2\}$ \\
%$\{y_1, \neg y_2, \neg x_1, x_2\}$&\hspace{0.2cm} &$ \{\neg y_1, y_2, \neg x_1,  x_2\}$
%\end{tabular}
%\end{nscenter}
%
This formula is easily shown to be false, but no steps are possible in D-Q-Resolution, hence D-Q-Resolution is not complete \cite{Balabanov201486}. 
Consider now the following modification of that formula where the universal variables are doubled:
\begin{gather} \label{eq:counterex2}
\forall x_1 \forall x_1'  \forall x_2 \forall x_2' \exists y_1(x_1, x_1') \exists y_2(x_2, x_2')
\\ \notag
\begin{tabular}{lll}
$\{y_1, y_2, x_1, x_1'\}$&\hspace{0.2cm}&$\{\neg y_1, \neg y_2, x_1, x_1'\}$\\
$\{y_1, y_2, \neg x_1, \neg x_1', \neg x_2, \neg x_2'\}$&\hspace{0.2cm} &$\{\neg y_1, \neg y_2, \neg x_1, \neg x_1', \neg x_2, \neg x_2'\}$ \\
$\{y_1, \neg y_2, \neg x_1, \neg x_1', x_2, x_2'\}$&\hspace{0.2cm} &$\{ \neg y_1, y_2, \neg x_1, \neg x_1',  x_2,  x_2'\}$. 
\end{tabular}
\end{gather}
The falsity of \eqref{eq:counterex2} follows from the fact that
its hypothetical Skolem model % $(f_{y_1},f_{y_2})$ 
would immediately yield a Skolem model for %\eqref{eq:counterex1}
the original formula
using assignments with $x_1=x_1'$, $x_2=x_2'$.
But there is no such model because the original formula is false.
However, since we have doubled the universal literals we cannot perform any generalised \qurc steps to begin a refutation. This technique of doubling literals was first used in \cite{BWJ14}.

Now we look at another portion of the calculi from Fig.~\ref{fig:overview},
namely the calculi that utilise merging. 
As a specific example we consider \lqrc and show that it is not sound when lifted to DQBF in the natural way.

% Since the generalisations \irmc, \lqurc and \lquprc admit the same proofs as in \lqrc,
% this shows that these are unsound as well. % (there is an unresolved issue of how \lquprc would handle a quantifier index in DQBF but since it is unsound regardless we do not try to answer this here).  

To do this we look at the condition of (L$\exists$R) given in Fig.~\ref{fig:allrules}.
Here instead of requiring $\ind(x)<\ind(u)$ as a condition for $u$ becoming merged, we require $u\notin Y_x$.
This is unsound as we show by the following DQBF:
\begin{nscenter}\begin{tabular}{cc}
		
		$\forall u \forall v \exists x(u) \exists y(v) \exists z(u,v).\,$&
		
		\begin{tabular}{lll}
			$\{x,v,z\}$&\hspace{0.2cm}&$\{\neg x, \neg v, z\}$\\
			$\{y, u, \neg z\}$&\hspace{0.2cm} &$\{\neg y, \neg u, \neg z\}$ 
		\end{tabular} 
	\end{tabular}
\end{nscenter}
%$(x\vee v \vee z) \land (\neg x\vee \neg v \vee z) \land (y\vee u \vee \neg z) \land (\neg y\vee \neg u \vee \neg z).$
Its truth is witnessed by the Skolem functions $x(u)=u$, $y(v)=\neg v$, and $z(u,v)=(u\wedge v) \vee (\neg u \wedge \neg v)$. 
However, the lifted version of \lqrc admits a refutation:
\begin{prooftree}
	\AxiomC{$\{x,v,z\}$}
	\AxiomC{$\{\neg x, \neg v, z\}$}
	\BinaryInfC{$\{v^*, z\}$}
	\AxiomC{$\{y, u , \neg z\}$}
	\AxiomC{$\{\neg y, \neg u, \neg z\}$}
	\BinaryInfC{$\{u^*,\neg z\}$}
	\BinaryInfC{$\{u^*, v^*\}$}
	\UnaryInfC{$\{u^*\}$}
	\UnaryInfC{$\bot$}
\end{prooftree}
This shows that \lqrc is unsound for DQBF. Likewise, since \irmc, \lqurc and \lquprc step-wise simulate \lqrc,
this proof would also be available, showing that these are all unsound calculi in the DQBF setting.

\section{A sound and complete proof system for DQBF}
\label{sec:dirc}

In this section we introduce the \dirc refutation system and prove its soundness and completeness for DQBF.
The calculus takes inspiration from \irc, a system for QBF \cite{BCJ14}, which in turn is inspired by first-order translations of QBF.
One such translation is to the \epr fragment, i.e., the universal fragment of 
first-order logic without function symbols of non-zero arity (this means we only allow constants).
We broaden this translation to DQBF and then bring this back down to \dirc in a similar way as in \irc.

We adapt annotated literals $l^\tau$ to DQBF, such that $l$ is an existential literal and $\tau$ 
is an annotation which is a partial assignment to universal variables in~$Y_x$.
In QBF, $Y_x$ contains all universal variables with an index lower than $x$,
and this is exactly the maximal range of the potential annotation to $x$ literals.
Thus our definition of annotated literals generalises those used in \irc.

\iffalse

\begin{definition} \label{def:restrict-dirc}
	Fix a DQBF $\Pi\cdot \psi$. Let $\tau$ be a partial assignment of the universal variables $Y$ to $\{0,1\}$
	and let $x$ be an existential variable. $\rest_x(\tau)$ is the assignment where $\domain(\rest_x(\tau))= \domain(\tau)\cap Y_x$ and $\rest_x(\tau)(u)=\tau(u)$. 
\end{definition}

\begin{definition} \label{def:instantiate-dirc}
	We define $\instantiate_\tau(C)$ to be the clause containing all the literals $l^{\rest_{\var(l)}(\sigma)}$, 
	where $l^\xi\in C$ and $\domain(\sigma)=\domain(\xi)\cup \domain(\tau)$ and $\sigma(u)=\xi(u)$ if $u\in \domain(\xi) $ and $\sigma(u)=\tau(u)$ otherwise.
\end{definition}

\fi

The definitions of $\rest$ and $\instantiate$ were defined for QBF, but with dependency already in mind.
With these definitions at hand we can now define the new calculus \dirc.  
Its rules are exactly the same as the ones for \irc stated in Fig.~\ref{fig:IRC}, but applied to DQBF.
\iffalse
\begin{figure}[t]
	\framebox{\parbox{\textwidth}
		{%\small
			\begin{prooftree}
				\AxiomC{}
				\RightLabel{(Axiom)}
				\UnaryInfC{$\comprehension{x^{\rest_x(\tau)}}{x\in C, x\text{ is existential}} $}
			\end{prooftree}
			$C$ is a non-tautological clause from the matrix. $\tau=\comprehension{0/u}{u\text{ is universal in }C}$, where the notation $0/u$ for literals $u$ is shorthand for $0/x$ if $u=x$ and $1/x$ if $u=\neg x$.
	\begin{prooftree}
		\AxiomC{$x^\tau\lor C_1 $ }
		\AxiomC{$\lnot x^\tau\lor C_2 $}
		\RightLabel{(Resolution)}
		\BinaryInfC{$C_1\cup C_2$}
		%\end{prooftree}
		\DisplayProof\hspace{1cm}
		%\begin{prooftree}
		\AxiomC{$C$}
		\RightLabel{(Instantiation)}
		\UnaryInfC{$\instantiate_\tau(C)$}
	\end{prooftree}
	\text{$\tau$ is an assignment to universal variables with $\range(\tau) \subseteq \{0,1\}$.}
	\caption{ The rules of \dirc.}
	\label{fig:DIR}
}}
\end{figure}
\fi

% \textbf{First-order translation.}
Before analysing \dirc further we present the translation of DQBF into \epr.
We use an adaptation of the translation described for QBF \cite{PAAR-2012:qbf2epr_A_Tool_for_Generating_EPR_Formulas_from_QBF},
which becomes straightforward in the light of the DQBF semantics based on Skolem functions.
The key observation is that for the intended two valued Boolean domain
the Skolem functions can be represented by predicates. 

To translate a DQBF $\Pi\cdot\psi$ we introduce on the first-order side
1) a predicate symbol $p$ of arity one and two constant symbols $0$ and $1$ to describe the Boolean domain,
2) for every existential variable $x \in X$ a % (distinct)
 predicate symbol $x$ of arity~$|Y_x|$, and
3) for every universal variable $y \in Y$ a % (distinct) 
first-order variable $\fovar{y}$.
% This means we only need a finite subset $\comprehension{u_y}{y\in Y} \subseteq U$
% of the countable set $U$ of the available first order variables.

Now we can define a translation mapping $t_\Pi$.
It translates each occurrence of an existential variable $x$ with dependencies $Y_x = \{y_1,\ldots,y_k\}$
to the atom $t_\Pi(x) = x(\fovar{y_1},\ldots,\fovar{y_k})$ (we assume an arbitrary but fixed
order on the dependencies which dictates their placement as arguments)
and each occurrence of a universal variable $y$ to the atom $t_\Pi(y) = p(\fovar{y})$.
The mapping is then homomorphically extended to formulas: 
$t_\Pi(\neg \psi) = \neg t_\Pi(\psi)$,
$t_\Pi( \psi_1 \lor \psi_2) = t_\Pi(\psi_1) \lor t_\Pi(\psi_2)$,
and $t_\Pi( \psi_1 \land \psi_2) = t_\Pi(\psi_1) \land t_\Pi(\psi_2)$.
This means a CNF matrix $\psi$ is mapped to a corresponding first-order CNF $t_\Pi(\psi)$.
As customary, the first-order variables of $t_\Pi(\psi)$ are assumed to be implicitly universally quantified
at the top level.

% The mapping satisfies the following.

\begin{lemma} % [\cite{PAAR-2012:qbf2epr_A_Tool_for_Generating_EPR_Formulas_from_QBF}] % why does it look ugly?
\label{lem:basecase}
A DQBF $\Pi\cdot\psi$ is true if and only if $t_\Pi(\psi) \land p(1) \land \neg p(0)$ is satisfiable.
\end{lemma}
\begin{proof}[Idea]
When the DQBF $\Pi\cdot\psi$ is true, this is witnessed by the existence of Skolem functions 
$F=\comprehension{f_x}{x\in X}$. On the other hand, if $t_\Pi(\psi) \land p(1) \land \neg p(0)$
is satisfiable then we can by Herbrand's theorem assume it has a Herbrand model $H$ over the base $\{0,1\}$.
We can naturally translate between one and the other by setting 
$f_x(\vec{v}) = 1 \text{ iff } x(\vec{v}) \in H$ for every $x \in X$ and $\vec{v} \in \{0,1\}^{|Y_x|}$.
The lemma then follows by structural induction over $\psi$.
\qed
\end{proof}

For the purpose of analysing \dirc, the mapping $t_\Pi$ is further extended to annotated literals:
$t_\Pi(x^\tau) = t_\Pi(x)\tau$ for an existential variable $x$.
Here we continue to slightly abuse notation and treat $\tau$,
an annotation in the propositional context,
as a first-order substitution over the corresponding translated variables in the first-order context
(recall point 3) above).
% (By a standard convention, a substitution acts as the identity mapping outside its explicit domain.)

% TODO: have an example sooner or later !

We aim to show soundness and completeness of \dirc by
relating it via the above translation to a first-order resolution calculus \for.
This calculus consists of 1) a lazy grounding rule: given a clause $C$ and a substitution $\sigma$ derive the instance $C\sigma$,
and 2) the resolution rule: given two clauses $C \cup \{l\}$ and $D \cup \{\neg l\}$, where $l$ is a first-order literal,
derive $C \cup D$. Note that similarly to propositional clauses, we understand 
first-order clauses as \emph{sets} of literals. Thus we do not need any explicit factoring rule.
% Could be left out:
Also note that we require the resolved literals of the two premises of the resolution rule to be 
equal (up to the polarity). Standard first-order resolution, which involves \emph{unification} of the resolved literals,
can be simulated in \for by combining the instantiation and the resolution rule. It is clear that \for is sound and complete for first-order logic.

Our argument for the soundness of \dirc is now the following.
Given $\pi= (L_1, L_2, \dots, L_\ell)$, a \dirc derivation of the empty clause $L_\ell = \bot$
from DQBF $\Pi\cdot\psi$, we show by induction that $t_\Pi(L_n)$ is derivable
from $\Psi = t_\Pi(\psi) \land p(1) \land \neg p(0)$ by \for for every $n\leq \ell$.
Because $t_\Pi(\bot)=\bot$ is unsatisfiable, so must $\Psi$ be, by soundness of \for
and therefore $\Pi\cdot\psi$ is false by Lemma~\ref{lem:basecase}.

We need to consider the three cases by which a clause is derived in \dirc.
First, it is easy to verify that \dirc instantiation by an annotation 
$\tau$ corresponds to \for instantiation by $\tau$ as a substitution,
i.e., $t_\Pi(\instantiate_\tau(C)) = t_\Pi(C)\tau.$
Also the \dirc and \for resolution rules correspond one to one in an obvious way.
Thus the most interesting case concerns the Axiom rule.

Intuitively, the Axiom rule of \dirc removes universal variables from a clause while
recording their past presence (and polarity) within the applied annotation $\tau$. We simulate
this step in \for by first instantiating the translated clause by $\tau$
and then resolving the obtained clause with the unit $p(1)$ and/or $\lnot p(0)$.
Here is an example for a DQBF prefix $\Pi = \forall u\, \forall v\, \forall w\, \exists x(u,v)\, \exists y(v,w)$:
\begin{prooftree}
	\AxiomC{$\{x , y , \neg u , v \}$  }
	\LeftLabel{(\dirc)}
	\UnaryInfC{$\{x^{1/u,0/v} , y^{0/v}\}$}
	\DisplayProof % \hspace{1cm}
	\AxiomC{$\{ \fopre{x}(\fovar{u},\fovar{v}) , \fopre{y}(\fovar{v},\fovar{w}) , \neg p(\fovar{u}) , p(\fovar{v})\}$}
\quad
	\LeftLabel{(\for)}
	\UnaryInfC{$\{ \fopre{x}(1,0) , \fopre{y}(0,\fovar{w})\}$}
\end{prooftree}

\begin{theorem} \label{thm:soundness}
       \dirc is sound.
\end{theorem}

We now show completeness. Let $\Pi\cdot\psi$ be a false DQBF and 
let us consider $\mathcal{G}(t_\Pi(\psi))$, the set of all ground instances of clauses in $t_\Pi(\psi)$.
Here, by a ground instance of a clause $C$ we mean the clause $C\sigma$
for some substitution $\sigma:\var(C)\rightarrow \{0,1\}$.
By the combination of Lemma~\ref{lem:basecase} and Herbrand's theorem, 
$\mathcal{G}(t_\Pi(\psi)) \land p(1) \land \neg p(0)$ is unsatisfiable
and thus it has a \for refutation. Moreover, by completeness of \emph{ordered}
resolution \cite{DBLP:books/el/RV01/BachmairG01}, we can assume that 1) the refutation 
does not contain clauses subsumed by $p(1)$ or $\neg p(0)$, and  
2) any clause containing the predicate $p$ is resolved on a literal containing $p$.
From this it is easy to see that any leaf in the refutation gives rise 
(in zero, one or two resolution steps with $p(1)$ or $\neg p(0)$)
to a clause $D = t_\Pi(C)$ where $C$ can be obtained by \dirc Axiom from a $C_0 \in \psi$.
The rest of the refutation consists of \for resolution steps 
which can be simulated by \dirc. Thus we obtain the following.

\begin{theorem} \label{thm:completeness}
	\dirc is refutationally complete for DQBF.
\end{theorem}

%\paragraph{Remarks.}
 Although one can lift the above argument with ordered resolution to show that
the set 
$\comprehension{t_\Pi(C)}{C \text{ follows by Axiom from some } C_0 \in \psi} $
is unsatisfiable for any false DQBF $\Pi\cdot\psi$, we have shown how to simulate 
\emph{ground} \for steps by \dirc. That is because a lifted \for derivation may contain 
instantiation steps which \emph{rename variables apart} for which a subsequent 
resolvent cannot be represented in \dirc. 
An example is the resolvent $\{\fopre{y}(\fovar{v}) , \fopre{z}(\fovar{v}')\}$
of clauses $\{\fopre{x}(\fovar{u}) , \fopre{y}(\fovar{v})\}$ and $\{\lnot \fopre{x}(\fovar{u}) , \fopre{z}(\fovar{v}')\}$
which is obviously stronger than the clause $\{\fopre{y}(\fovar{v}) , \fopre{z}(\fovar{v})\}$.
However, only the latter has a counterpart in \dirc. 

We also remark that in a similar way we can also lift to DQBF the QBF calculus \ecalculus from \cite{JM15}.
It is easily verified that the simulation of \ecalculus by \irc shown in \cite{BCJ14} directly transfers from QBF to DQBF.
Hence Theorem~\ref{thm:soundness} immediately implies the soundness of \ecalculus lifted to DQBF.  
Moreover, because all ground instances are also available in \ecalculus lifted to DQBF,
this system is also complete as can be shown by repeating the argument of Theorem~\ref{thm:completeness}.
  
%\begin{corollary}
%  \ecalculus is sound and complete for DQBF.
%\end{corollary}
%\begin{small}
\subsubsection*{Acknowledgments.}
This research was  supported by grant no.\ 48138 from the John Templeton Foundation, 
EPSRC grant EP/L024233/1, and 
a Doctoral Training Grant from the EPSRC (2nd author).

Martin Suda was supported by the EPSRC grant EP/K032674/1 and the ERC Starting Grant 2014 SYMCAR 639270.
%\end{small}

\clearpage
\bibliographystyle{splncs03}
\bibliography{compl,refs}
\end{document}